\documentclass[journal]{IEEEtran}
\usepackage{mathrsfs}
\usepackage{amsmath}
\usepackage{graphicx}
\usepackage{amssymb}
\usepackage{enumerate}
\usepackage{longtable,tabularx,float}
\usepackage{cite}
\usepackage{mathcomp}
\usepackage{supertabular}
\usepackage{stmaryrd}
\usepackage{color}
\usepackage{url}
\usepackage{amsthm}
\newtheorem{definition}{Definition}[section]
\newtheorem{lemma}{Lemma}[section]
\newtheorem{theorem}{Theorem}[section]

\newtheorem{example}{Example}[section]
\newtheorem{construction}{Construction}[section]

\newcommand{\A}{\mathcal {A}}
\newcommand{\B}{\mathcal {B}}
\newcommand{\C}{\mathcal {C}}
\newcommand{\G}{\mathcal {G}}

\newcommand{\bbZ}{{\mathbb Z}}
\newcommand{\bbF}{{\mathbb F}}
\newcommand{\vu}{{\sf u}}

\newcommand{\vx}{{\sf x}}
\newcommand{\vy}{{\sf y}}

\begin{document}

\title{Optimal $q$-Ary Error
  Correcting/All Unidirectional Error Detecting Codes}

\author{Yeow~Meng~Chee,~{\it Senior~Member,~IEEE},~and~Xiande~Zhang
        \thanks{Y. M. Chee ({\tt ymchee@ntu.edu.sg})
        is with Division of Mathematical Sciences, School of
Physical and Mathematical Sciences, Nanyang Technological
University, Singapore 637371.}

\thanks{X. Zhang ({\tt drzhangx@ustc.edu.cn}) is with School of Mathematical Sciences,
University of Science and Technology of China, Hefei, 230026, Anhui, China.  The research of X. Zhang is supported by NSFC under grant 11771419 and 11301503, and by ``the  Fundamental
Research Funds for the Central Universities''.}
\thanks{Copyright (c) 2017 IEEE. Personal use of this material is permitted.  However, permission to use this material for any other purposes must be obtained from the IEEE by sending a request to pubs-permissions@ieee.org.
 }
}
\maketitle

\begin{abstract}
\boldmath Codes that can correct up to $t$ symmetric errors and
detect all unidirectional errors, known as $t$-EC-AUED codes, are
studied in this paper. Given positive integers $q$, $a$ and $t$, let
$n_q(a,t+1)$ denote the length of the shortest $q$-ary $t$-EC-AUED
code of size $a$. We introduce combinatorial
constructions for $q$-ary $t$-EC-AUED codes via one-factorizations of complete graphs,
and concatenation of MDS codes and codes from resolvable set systems. Consequently, we
determine the exact values of $n_q(a,t+1)$ for several new infinite
families of $q,a$ and $t$.
\end{abstract}


\section{Introduction}

\IEEEPARstart{C}{lassical}  error control codes have been designed
for use on binary {\em symmetric channels}, i.e., both $1\rightarrow
0$ and $0\rightarrow 1$ errors can occur during transmission.
However, errors in some VLSI and optical systems are asymmetric in
nature \cite{ConstantinRao:1979,Pierce:1978}, where the error
probability from $1$ to $0$ is significantly higher than that from
$0$ to $1$. Practically we can assume that only one type of errors
can occur in those systems. These errors are called {\em asymmetric
errors}.

Different from asymmetric errors, {\em
unidirectional errors} can be caused by certain faults in digital devices, where both $1\rightarrow
0$ and $0\rightarrow 1$ type of errors are possible, but in any particular word all the errors are of the same type. Digital units that
produce unidirectional errors as a consequence of internal failure
are  data transmission systems, magnetic recording mass memories,
and LSI/VLSI circuits such as ROM memories \cite{BlaumVan:1989}. The
number of random errors caused by these failures is usually limited,
while the number of unidirectional errors can be large. For this
reason, it is useful to consider codes that are capable of
correcting a relatively small number of  random errors  and
detecting  any number of unidirectional errors.   Considerable
attention  has been  paid  to this problem, see for example
\cite{BlaumVan:1989,BoiFraTil1990,NaydeTor2009,CheeLing:2009,zhang1993lym,zhang1993construction,bruck1992new,lin1993constant,al2000another,nikolos1986systematic,pradhan1980new,lin1988theory,tao1988efficient,kundu1990symmetric}.

In 1973, Varshamov introduced a $q$-ary asymmetric channel
\cite{Varshamov:1973}, where the inputs and outputs of the channel
are sequences over the $q$-ary alphabet $R = \{ 0, 1, \ldots , q-1
\}$. If the symbol $i$ is transmitted then the only symbols which
the receiver can get are $ i, i + 1, \ldots , q - 1 $. We say that
the type of this error is {\em increasing}. Naturally, we have another
type of error which is {\em decreasing}.  The {\em $q$-ary unidirectional
channel} is the channel on which all errors within a codeword are of
the same type (all increasing or all decreasing). Recent work on
$q$-ary unidirectional errors can be found in
\cite{NaydeTor2009,ElariefBose:2010,Fu:2004,ahlswede2006q,ahlswede2002unidirectional} for example.

In this paper, we study constructions for $q$-ary
codes which can correct up to $t$ symmetric errors and detect all
unidirectional errors (known as $t$-EC-AUED codes). We are only interested in codes that are
optimal when considering the shortest lengths for
given sizes. Let $n_q(a,t+1)$ denote the length of the shortest
$q$-ary $t$-EC-AUED code of size $a$.
We introduce several combinatorial
constructions for $t$-EC-AUED codes and
determine the exact values of $n_q(a,t+1)$ for new infinite
families of $q,a$ and $t$.

Our main results are as follows:

 \begin{enumerate}[(i)]
 \item determining values of $n_3(a,t+1)$ for $a\leq 12$ and  all $t$;
\item for integers $k\geq 2$, $n_k(a,k-1)=2k-1$ with $k+1\leq a\leq 2k-1$; if $k$ is odd, then $n_k(2k,k-1)=2k-1$;
\item for prime powers $q\geq 2$, $n_q(a,q)=2q+2$ with $2q-1\leq a\leq q^2$;
\item given positive integers $s$, $\lambda\geq1$ and $\alpha\geq 2$, \[n_{q}(sn,
T)=2\lambda s(n-1)/(\alpha-1)\] for all sufficiently large $n$
satisfying $sn\equiv 0 \pmod{\alpha}$ and $\lambda s(n-1)\equiv 0
\pmod{\alpha-1}$, where $T=\frac{\lambda s(n-1)}{\alpha -1}-\lambda$
and $q=\frac{sn}{\alpha}$;
\item given positive integers $\lambda\geq1$ and $\alpha\geq 2$, \[n_{q}(\alpha q, T)=2\lambda(\alpha q-1)/(\alpha-1)-2\]
for all sufficiently large $q$ satisfying $\lambda(\alpha q-1)\equiv
0 \pmod{\alpha-1}$, where $T=\frac{\lambda\alpha (q-1)}{\alpha -1}-1$.
\item values for several other families of $n_q(a,t+1)$ are stated in Table I.
\end{enumerate}
Previously, only values of $n_2(a,t+1)$ for $a\leq 14$ and $n_3(a,t+1)$ for $a\leq 9$
were known by  \cite{NaydeTor2009}.

Our paper is organized as follows. In Section \ref{pre}, we introduce necessary notation and briefly describe the problem status. Section~\ref{fac} gives a construction of optimal EC-AUED codes from near one-factorizations, where result (ii) is obtained. In Section~\ref{conca}, we apply concatenation method to various codes with good Hamming distance to get results (iii)-(vi). In Section \ref{ter}, we improve results in \cite{NaydeTor2009} and determine completely $n_3(a,t+1)$ for $a\leq 12$ and all $t$, which is our result (i). Finally, a conclusion is given in Section \ref{con}.

\section{Preliminaries}\label{pre}

Necessary and sufficient conditions for correcting and detecting
errors of each of the three types, symmetric, asymmetric and
unidirectional, are known in \cite{Weber:1985,Weber92nece}. To state these
conditions, we need some necessary notation.

 Let $X$ be a finite
set, and $R^X$ denote the set of vectors of length $|X|$, where each
component of a vector $\vu\in R^X$ has value in $R$ and is indexed
by an element of $X$, that is, $\vu=(\vu_x)_{x\in X}$, and $\vu_x\in
R$ for each $x\in X$.
For $\vx,\vy\in R^X$, let $N(\vx,\vy)$ denote
the number of positions $i$ where $\vx_i>\vy_i$. If $N(\vy,\vx)=0$, then the
vector $\vx$ is said to {\em cover} the vector $\vy$ and we write $\vx\geq
\vy$. If $\vx\geq \vy$ or $\vy\geq \vx$ the vectors $\vx$ and $\vy$ are said to be
{\em ordered}, otherwise they are {\em unordered}.

A {\em code} is a set $\C\subseteq R^X$ for some $X$. The elements of $\C$ are called {\em
codewords}. A code
is called a {\em $t$-EC-AUED code} if it is able to correct up to
$t$ symmetric errors and detect all unidirectional errors. Clearly a
code is $0$-EC-AUED if any pair of codewords are unordered. For
general $t$, a characterization of when a code is a $t$-EC-AUED code
is known as follows.

\begin{theorem}\cite{BoseRao:1982}\label{aued} A code $\C$ is a $t$-EC-AUED code if and only if $N(\vx,\vy)\geq t+1$ and $N(\vy,\vx)\geq t+1$, for all
distinct $\vx,\vy\in \C$.
\end{theorem}

Define the {\em asymmetric distance} of two vectors $\vx$ and $\vy$ as
$d_{as}(\vx, \vy) = \min\{N(\vx, \vy), N(y, \vx)\}$. Then codes with minimum asymmetric distance $T$
are $(T - 1)$-EC-AUED codes.
Let $n_q(a,T)$ denote the length of the shortest $q$-ary $(T-1)$-EC-AUED code
 of size $a$. We say that a $q$-ary $(T-1)$-EC-AUED code of length
$n_q(a,T)$ and size $a$ is {\em optimal}.

The lower bound derived by B\"oinck and van Tilborg
\cite{BoiFraTil1990} for the length of  binary $(T-1)$-EC-AUED codes is
\[n_2(a,T)\geq \left\lceil (4-\frac{2}{\left\lceil a/2\right\rceil})T\right\rceil.\] In the same paper, they show that if
$n_2(a,T)=  (4-\frac{2}{\left\lceil a/2\right\rceil})T$ holds, then the code
must be a constant weight code; and if  further $a\equiv 0 \pmod
4$, then $T$ must be divisible by  $a/2$.

For non-binary codes, the lower bound of $n_q(a,T)$ was generalized
in \cite{NaydeTor2009}.
 \begin{theorem}\label{lbound}\cite{NaydeTor2009}
 $n_q(a,T)\geq GBT_q(a,T)$,
 where \[GBT_q(a,T)=\left\lceil\frac{2a(a-1)T}{a(a-\alpha)-(a-\alpha
 q)(\alpha+1)}\right\rceil\] and $\alpha =\left\lfloor a/q\right\rfloor$.
 \end{theorem}

The function $GBT$ has a property that for all $\mu\geq 0$,
$GBT_q(q\mu+(q-1),T)=GBT_q(q\mu+q,T)$. That is,  by deleting one
codeword from the optimal code of size $q\mu+q$, we obtain an
optimal code of size $q\mu+(q-1)$.

\begin{lemma}\label{deleter}\cite{NaydeTor2009} If  $n_q(q\mu+q,T)=GBT_q(q\mu+q,T)$, then
$n_q(q\mu+q-1,T)=GBT_q(q\mu+q-1,T)$.
\end{lemma}

In fact, Lemma~\ref{deleter} can be extended whenever $n_q(a,T)=GBT_q(a,T)$ and $GBT_q(a',T)=GBT_q(a,T)$, for $a'<a$.

Both the B\"oinck-van Tilborg bound and $GBT$ bound are closely  related to Plotkin bound, where the codes achieve the bounds when each symbol occurs almost the same number of times in a fixed position. In such cases, concatenating short codes is a very useful method to construct optimal long codes \cite{MackSeb84ternarycodes}. As stated in the following lemma, Naydenova and Kl{\o}ve
\cite{NaydeTor2009} showed that optimal
$t$-EC-AUED codes could be obtained by concatenating two optimal short codes for fixed $q$, $a$ and general $T$.

\begin{lemma}\label{recur}\cite{NaydeTor2009} If $n_q(a,T_1)=GBT_q(a,T_1)$,
$n_q(a,T_2)=GBT_q(a,T_2)$, and
\[GBT_q(a,T_1)+GBT_q(a,T_2)=GBT_q(a,T_1+T_2),\] then \[n_q(a,T_1+T_2)=GBT_q(a,T_1+T_2).\]
\end{lemma}

By Theorem~\ref{lbound}, we have $n_q(a,T)\geq 2T$ if $q\geq a$. In
fact, $n_q(a,T)= 2T$ in this case since the $a\times 2T$ array
formed by $T$ column vectors $(1,2,\ldots,a)$ and $T$ column vectors
$(a,a-1,\ldots,1)$ is an optimal code. From now on, we assume that
$q<a$.

In  \cite{NaydeTor2009}, the values of $n_2(a,T)$ for $a\leq 14$ and
all $T$ have been determined by direct constructions and the
B\"oinck-van Tilborg bound. For ternary case, they constructed some
optimal codes up to size $9$. We summarize their results for ternary
codes as below.

\begin{lemma}\label{known}\cite{NaydeTor2009} For $T\geq 1$, we have
 \begin{enumerate}[(i)]
\item $n_3(a,T)=GBT_3(a,T)$ for $a\in \{4,5,6\}$;
\item $\left\lceil21T/8\right\rceil\leq n_3(7,T)\leq \left\lceil8T/3\right\rceil$;
\item $n_3(a,T)= \left\lceil8T/3\right\rceil=GBT_3(a,T)$
 for $a\in\{8,9\}$ and $T\not\equiv 1 \pmod 3$.
\end{enumerate}
\end{lemma}

\section{A Construction from One-Factorizations}\label{fac}

In this section, we give a construction of optimal $q$-ary
$t$-EC-AUED codes based on one-factorizations, which yield our main result (ii)  by the extension of Lemma~\ref{deleter}.
%
For integers $m\leq n$,
the set of integers $\{m,m+1,\ldots,n\}$ is denoted by $[m,n]$. When $m=1$, the
set $[1,n]$ is further abbreviated to $[n]$.

Let $a=2k-1$. The ring $\bbZ/a\bbZ$ is denoted by $\bbZ_a$. Let
$K_a$ be a complete graph with vertex set $\bbZ_a$. For each $j\in
\bbZ_a$, take
\begin{align}\label{factor}
T_j=\{\{t+j,-t+j\}: 1\leq t\leq k-1\},
\end{align}
 where the addition is in $\bbZ_a$. Then $\{T_j:j\in
\bbZ_a\}$ is a {\em near one-factorization of $K_a$}. Each $T_j$ is
a {\em near one-factor} which misses the vertex $j$.

\begin{construction}
\label{con1} For each $k\geq 2$, construct a $(2k-1)\times (2k-1)$ array $A$ over
$[0,k-1]$, where rows and columns are indexed by
$\bbZ_{2k-1}$. For a cell in the $i$th row and the $j$th column,
let $A_{i,j}=0$ if $i=j$ and $A_{i,j}=x$ if $i\in\{x+j,-x+j\}$.  Let $\A$ be the collection of rows of $A$.
\end{construction}

\begin{theorem}
\label{nearfactorization}$n_k(2k-1,k-1)=2k-1$ for all integers $k\geq
2$.
\end{theorem}

\begin{proof} By Theorem~\ref{lbound}, we have $n_k(2k-1,k-1)\geq 2k-1$.
It suffices to prove that the code $\A$ constructed in Construction~\ref{con1} has minimum asymmetric distance $k-1$.

 For any two
rows $\vx$ and $\vy$ of $A$ indexed by $i_1$ and $i_2$ respectively, we claim that  $N(\vx,\vy)\geq k-1$. In fact,
by the definition of near one-factorization, there exists a column indexed by
$j_0$ and an element $x_0\in [k-1]$ such that
$\{i_1,i_2\}=\{x_0+j_0,-x_0+j_0\}\in T_{j_0}$, i.e.,
$A_{i_1,j_0}=A_{i_2,j_0}=x_0$. Without loss of generality, assume
that $i_1=x_0+j_0$ and $i_2=-x_0+j_0$. Thus for each $y\in
[k-1]$, we have \[i_1=-y+j_0+(x_0+y)\] and
\[i_2=y+j_0+(-x_0-y)\] in $\bbZ_{2k-1}$.
By the construction of $A$, we have $A_{i_1,j_0-y}=A_{i_2,j_0+y}$
which equals  $x_0+y$ or $-(x_0+y)$ whoever falls in
$[k-1]$. So for each $y\in [k-1]$, if
$A_{i_1,j_0-y}>A_{i_2,j_0-y}$, we must have
$A_{i_1,j_0+y}<A_{i_2,j_0+y}$. Thus $N(\vx,\vy)=N(\vy,\vx)= k-1$.
\end{proof}
\begin{example}\label{k23}Let $k=2$ and $k=3$. Then applying Construction~\ref{con1} and Theorem~\ref{nearfactorization}
gives an optimal binary $0$-EC AUED codes of size three, and an optimal ternary $1$-EC AUED codes of size five.
$$\begin{array}{ccc}
 0& 1 &1  \\
 1&0  &1  \\
  1& 1 & 0 \\
\end{array}$$

$$\begin{array}{ccccc}
 0&  1& 2 &2 & 1\\
 1& 0 & 1 &2 & 2\\
 2 & 1 & 0 & 1& 2\\
 2  &2  & 1 & 0&1 \\
  1  & 2 & 2 &1 & 0\\
\end{array}$$
\end{example}

Remark 1: The array $A$ from
Construction~\ref{con1} has extra property that in each row,
each nonzero element occurs twice and the zero element occurs
exactly once.

\begin{construction}
\label{con2} For odd integers $k\geq 3$, let $B$ be a $(2k-1)\times (2k-1)$ array with the entry  $B_{i,j}\equiv A_{i,j}+(k-1)/2
\pmod k$, where $A$ is the array from
Construction~\ref{con1}. Let $\vu$ be a
vector of length $2k-1$ with all entries being $(k-1)/2$. Denote $\B$ the collection of rows of $B$ and let $\B'=\B\cup\{\vu\}$.
\end{construction}

\begin{theorem}
\label{factorization}$n_k(2k,k-1)=2k-1$ for all odd integers $k\geq
3$.
\end{theorem}
\begin{proof} The lower bound can be checked by Theorem~\ref{lbound}. For the upper
bound, we only need to show that $\B$ in Construction~\ref{con2} is a $(k-2)$-EC-AUED code by Remark 1. This follows from the fact that
$B_{i_1,j_0-y}=B_{i_2,j_0+y}$ for each $y\in [1,k-1]$ as in the proof of Theorem~\ref{nearfactorization}.
\end{proof}

\begin{example}\label{k3}Applying Construction~\ref{con2} and Theorem~\ref{factorization} with $k=3$ gives  an optimal ternary $1$-EC AUED codes of size six.
$$\begin{array}{ccccc}
 1&  1& 1 &1 & 1\\
 1&  2& 0 &0 & 2\\
 2& 1 & 2 &0 & 0\\
 0 & 2 & 1 & 2& 0\\
 0  &0  & 2 & 1&2 \\
  2  & 0 & 0 &2 & 1\\
\end{array}$$
\end{example}

\section{A Construction by Concatenation}\label{conca}
We first give a simple but very useful construction of EC-AUED codes by concatenation. As mentioned in Section II, this method has been widely used to construct codes achieving Plotkin type bounds. For any $q$-ary word $c=(c_1,c_2,\ldots,c_n)$, let $q-1-c:=(q-1-c_1,q-1-c_2,\ldots,q-1-c_n)$ and $c|(q-1-c)=(c_1, \ldots, c_n, q - 1 -c_1, \ldots, q - 1 -c_n)$. For any two words $\vx$ and $\vy$, the Hamming distance of $\vx$ and $\vy$, denoted by $d_H(\vx,\vy)$, is the number of positions $i$ such that $\vx_i=\vy_i$. It is obvious that $d_H(\vx,\vy)=N(\vx,\vy)+N(\vy,\vx)$.

\begin{lemma}
\label{concatenation}Let $C$ be a $q$-ary code of length $n$ with minimum Hamming
distance $d$. Then
$\{c |(q-1-c): c \in C\}$
is a $q$-ary $(d - 1)$-EC-AUED code of length $2n$ with $|C|$ words.
\end{lemma}
\begin{proof} For any two words $\vx=c |(q-1-c)$ and $\vy=c' |(q-1-c')$, since $N(q-1-c,q-1-c')=N(c',c)$, we have $N(\vx,\vy)=N(c,c')+N(q-1-c,q-1-c')=N(c,c')+N(c',c)=d_H(c,c')\geq d$. It's similar that $N(\vy,\vx)\geq d$.
\end{proof}

By Lemma~\ref{concatenation}, we can construct good EC-AUED codes from codes with large Hamming distance.

\begin{theorem}
\label{calay}$n_q(q^2,q)=2q+2$ for all prime powers $q$.
\end{theorem}
\begin{proof} The lower bound is checked by Theorem~\ref{lbound}. The upper bound is obtained by applying Lemma~\ref{concatenation} to  $q$-ary MDS codes of length $q + 1$ and size $q^2$ with minimum Hamming distance $q$ \cite[Chapter 5]{huffman2010fundamentals}.
\end{proof}

\subsection{Constructions from Set Systems}\label{rsys}

A {\em set system} is a pair ${\frak S}=(X,\A)$, where $X$ is a
finite set of {\em points} and $\A\subseteq 2^X$. Elements of $\A$
are called {\em blocks}. The {\em order} of $\frak S$ is the number
of points in $X$, and the {\em size} of $\frak S$ is the number of
blocks in $\A$. Let $K$ be a set of positive integers. A set system
$(X,\A)$ is {\em $K$-uniform} if $|A|\in K$ for all $A\in\A$. A {\em
parallel class} of a set system $(X,\A)$ is a set ${\cal
P}\subseteq\A$ that partitions $X$. A {\em resolvable set system} is
a set system whose set of blocks can be partitioned into parallel
classes. We refer the readers to \cite{GeMiao:HCD} for other related concepts in combinatorial design theory.

\begin{definition}
Let $(X,\A)$ be a $\{k\}$-uniform set system of order $n$. Then it
is an {\em $(n,k,\lambda)$-packing} if each pair of $X$ occurs in at
most $\lambda$ blocks of $\A$.
\end{definition}

Given a resolvable
$(qk,k,\lambda)$-packing of  $n$ parallel classes, arbitrarily order
the $q$ blocks in each parallel class by elements in $[0,q-1]$. Define an $qk\times n$  $q$-ary matrix $A$ by indexing
each column by a parallel class and each row by a point of the
packing. For each parallel class, the corresponding column has the
symbol $i$ in the rows indexed by the points in the $i$th block. Since each pair of points occurs in at most $\lambda$ blocks, the rows of $A$ form a $q$-ary code of Hamming distance at least $n-\lambda$. Note that this correspondence is the one used by Semakov and Zinoviev \cite{Semakov68} to show the equivalence between {\em equidistant codes} and RBIBDs. Recently, this method is used again to construct optimal {\em equitable symbol weight codes},  see for example \cite{chee2012optimal,dai2015two}.
 By applying Lemma~\ref{concatenation} to this equivalence, we have the following result.

\begin{lemma}
\label{const}Suppose that there exists a resolvable
$(a,k,\lambda)$-packing, which has $n$ parallel classes each
consisting of $q$ blocks, $q=a/k$. Then there exists a $q$-ary
$t$-EC-AUED code of size $a$ and length $2n$ with $t=n-\lambda-1$.
\end{lemma}

Next, we apply Lemma~\ref{const} to some concrete combinatorial
objects to determine the values of $n_q(a,T)$.

\begin{definition}
Let $(X,\A)$ be a $\{k\}$-uniform set system and let $\G$ be a
partition of $X$ into subsets, called {\em groups}. The triple
$(X,\G,\A)$ is a {\em group divisible design} (GDD) when every
2-subset of $X$ not contained in a group is contained in exactly
$\lambda$ block, and $|A\cap G|\leq 1$ for all $A\in\A$ and
$G\in\G$.
\end{definition}

We denote such a GDD $(X,\G,\A)$ by $(k,\lambda)$-GDD. It is obvious
that a $(k,\lambda)$-GDD $(X,\G,\A)$ is an $(n,k,\lambda)$-packing
with $n=|X|$.  The {\em type} of a GDD $(X,\G,\A)$ is the multiset
$\langle |G| : G\in\G\rangle$. When more convenient, the exponential
notation is used to describe the type of a GDD: a GDD of type
$g_1^{t_1} g_2^{t_2} \cdots g_s^{t_s}$ is a GDD where there are
exactly $t_i$ groups of size $g_i$, $i\in[s]$. When a GDD is
resolvable, we denote it by RGDD. A $(k,\lambda)$-GDD of type $1^n$
is called a {\em balanced incomplete block design}, denoted
 by BIBD$(n,k,\lambda)$ and RBIBD$(n,k,\lambda)$ when it is resolvable.

\begin{theorem}\label{rgdd} \cite{ChanDukes:2013}
Fix integers $g,\lambda\geq 1$ and $k\geq 2$. There exists an
integer $u_0(g,k)$ such that for all $u\geq u_0$,  a
$(k,\lambda)$-RGDD of type $g^u$ exists if and only if $\lambda g(u-1)
\equiv 0 \pmod{ k- 1}$ and $gu \equiv 0 \pmod{ k}$.
\end{theorem}

\begin{lemma}
\label{fromgdd}Suppose that there exists an $(\alpha,\lambda)$-RGDD
of type $s^n$, such that $2\lambda(s-1)<sn-\alpha$. Then $n_{q}(sn,
T)=2\lambda s(n-1)/(\alpha-1)$, where $T=\frac{\lambda
s(n-1)}{\alpha -1}-\lambda$ and $q=\frac{sn}{\alpha}$.
\end{lemma}
\begin{proof}It is easy to check that when $2\lambda(s-1)<sn-\alpha$, we have $n_{q}(sn,
T)\geq 2\lambda s(n-1)/(\alpha-1)$ by Theorem~\ref{lbound}. The
equality could be obtained by Lemma~\ref{const} and the fact that
the given GDD is a resolvable $(sn,\alpha,\lambda)$-packing.
\end{proof}

\begin{theorem}
\label{asymprgdd}Given positive integers $s$, $\lambda\geq1$ and $\alpha\geq 2$, $n_{q}(sn, T)=2\lambda s(n-1)/(\alpha-1)$ for all
sufficiently large $n$ satisfying $sn\equiv 0 \pmod{\alpha}$ and
$\lambda s(n-1)\equiv 0 \pmod{\alpha-1}$, where $T=\frac{\lambda
s(n-1)}{\alpha -1}-\lambda$ and $q=\frac{sn}{\alpha}$.
\end{theorem}
\begin{proof}For fixed $s,\alpha$ and $\lambda$, we have $2\lambda(s-1)<sn-\alpha$ for sufficiently large $n$.
Hence the conclusion follows by Lemma~\ref{fromgdd}  and the
asymptotic existence of $(\alpha,\lambda)$-RGDD of type $s^n$ in
Theorem~\ref{rgdd}.
\end{proof}


\begin{lemma}
\label{delonecol}If there exists an RBIBD$(\alpha q,\alpha,\lambda)$
with $q\geq 3$, then $n_{q}(\alpha q, T)=2\lambda(\alpha
q-1)/(\alpha-1)-2$, where $T=\frac{\lambda\alpha (q-1)}{\alpha
-1}-1$.
\end{lemma}
\begin{proof}Delete one parallel class from the RBIBD$(\alpha
q,\alpha,\lambda)$ to get a resolvable $(\alpha
q,\alpha,\lambda)$-packing. Then apply Lemma~\ref{const}.
\end{proof}

\begin{theorem}
\label{asymp}Given positive integers $\lambda\geq1$ and $\alpha\geq 2$, $n_{q}(\alpha q, T)=2\lambda(\alpha q-1)/(\alpha-1)-2$
for all sufficiently large $q$ satisfying $\lambda(\alpha q-1)\equiv
0 \pmod{\alpha-1}$, where $T=\frac{\lambda\alpha (q-1)}{\alpha -1}-1$.
\end{theorem}
\begin{proof}By the asymptotic existence RBIBD$(\alpha
q,\alpha,\lambda)$ in Theorem~\ref{rgdd}.
\end{proof}

In Table \ref{t1}, we give some examples of exact values of $n_q(a,T)$
determined by Lemmas~\ref{fromgdd}, \ref{delonecol} and the extension of Lemma~\ref{deleter}. The existence of combinatorial
objects used in this table can be found in \cite{GeMiao:HCD}.

\begin{table*}
\center\caption{Values of $n_q(a,T)$ from Lemmas~\ref{fromgdd} and \ref{delonecol}}\label{t1}
\[\begin{array}{c|c|c|c|c|c}
\hline
k & q & a & T & n_q(a,T) &  \text{Apply Lemma~\ref{fromgdd} with}\\
\hline \geq 1& 2k+1&[4k+1, 6k+3] & 3k& 6k+2 &
\text{RBIBD$(6k+3,3,1)$}\\
\hline \geq 3& 2k&[4k-1, 6k] & 3k-2& 6k-2 &
\text{$(3,1)$-RGDD of type $2^{3k}$}\\
\hline \geq 1& 3k+1 &[6k+1, 12k+4] & 4k& 8k+2 &
\text{RBIBD$(12k+4,4,1)$}\\
\hline \geq 2& 3k &[6k-1, 12k] & 4k-2& 8k-2 &
\text{$(4,1)$-RGDD of type $3^{4k}$}\\
\hline \geq 58& 3k+2 &[6k+3, 12k+8] & 4k+1& 8k+4 &
\text{$(4,1)$-RGDD of type $2^{6k+4}$}\\
 \hline
 \hline
 k & q & a & T & n_q(a,T) &  \text{Apply Lemma~\ref{delonecol} with}\\
 \hline \geq 2& 2k+1&[4k+1, 6k+3] & 3k-1& 6k+2 &
\text{RBIBD$(6k+3,3,1)$}\\
\hline \geq 2& 3k+1 &[6k+1, 12k+4] & 4k-1& 8k+2 &
\text{RBIBD$(12k+4,4,1)$}\\
\hline
\end{array}\]
\end{table*}

Before closing this section, we note that for two
binary vectors $\vx$ and $\vy$ of equal number of $1$'s, we have that $N(\vx, \vy) = N(\vy, \vx) =\frac{1}{2}d_H(\vx, \vy)$. So a binary
constant weight code with minimum Hamming distance $d$ is a $(\frac{d}{2}-1)$-EC AUED code. It's well known that the rows of the incidence matrix of a BIBD form a binary constant weight code. In Table \ref{t2}, we give two examples of equivalent objects for optimal
binary EC AUED codes. However, the
existence of corresponding BIBDs used in Table II
is very rare by referring to \cite{IoninTrung:HCD}.

\begin{table*}
\center\caption{Equivalent Objects for Binary Codes}\label{t2}
\[\begin{array}{c|c|c|c}
\hline
a & T & n_2(a,T)&  \text{Optimal codes equivalent to} \\
\hline 4k+3 & k+1& 4k+3 &
\text{BIBD$(4k+3,2k+1,k)$'s}\\
\hline 2k & k& 4k-2 &
\text{BIBD$(2k,k,k-1)$'s}\\
 \hline
\end{array}\]
\end{table*}

\section{Optimal Ternary $t$-EC-AUED Codes}\label{ter}
In this section, we give some direct constructions of optimal
ternary $t$-EC-AUED codes up to size $12$. For some of the codes we
search directly by computer, but when the length becomes big, the
searching space will be huge. In this case,  we map each ternary
code to be a resolvable set system as in Section IV. Suppose there
is a ternary $t$-EC-AUED code of size $a$ and length $n$. Let the
rows be indexed by $X$ of size $a$, then for each column, we obtain
three blocks by collecting all the indices of rows with same
entries. Thus we get a resolvable set system of order $a$ and size
$3n$, where each parallel class has three blocks. Conversely, we can
get the corresponding ternary code from such a resolvable set
system. However, to
ensure that the code is a $t$-EC-AUED code, the set system must satisfy extra conditions, which are not easy to be characterized.

In the following constructions, if we construct a resolvable set
system instead of the ternary code, we list the three blocks in each
parallel class in order, for which the entries in the corresponding
rows will be assigned to  $0$, $1$ and $2$ respectively. If we list
the optimal code itself, we usually denote $C_T$ the optimal
$(T-1)$-EC-AUED code.

Further, in design theory, people usually equip the desired designs
with some group structures to reduce the search space.
What they do is try to find a partial result, which can be developed
to the complete desired design by using the group structure. For example,
 if a block $B$ is developed by $\bbZ_n$, then $B_i$, $i\in \bbZ_n$ are
  obtained such that $B_i=\{b+i: b\in B\}$. We
will apply this idea to some of our constructions.

\begin{lemma}
\label{q3t8a7}$n_{3}(7,8)=21$.
\end{lemma}
\begin{proof}Let $X=\bbZ_7$. The following nine blocks form three
parallel classes, where each row is a parallel class. Develop them to $21$  parallel classes
by $\bbZ_7$ and keep the order of blocks in each parallel class.
Then one can check this resolvable set system gives a ternary
$7$-EC-AUED code of size $7$ and length $21$, which is optimal by the fact that
$GBT_{3}(7,8)=21$. In fact, the codewords are the rows of the $7 \times 21$ matrix $M = (A|B|C)$ where
$A$, $B$ and $C$ are circulant matrices with rows indexed by $\bbZ_7$. The leftmost column of $A$ has zeroes in rows $0,1,2$, ones in rows $3
$ and $6$, and twos in rows $4$ and $5$; the leftmost column of $B$ has zeroes in rows $0,2,4$, ones in rows $5
$ and $6$, and twos in rows $1$ and $3$; and the leftmost column of $C$ has zeroes in rows $0,1,4$, ones in rows $3
$ and $5$, and twos in rows $2$ and $6$.
\[\begin{array}{lll}
    \{0,1,2\},&\{3,6\},&\{4,5\}\\
\{0,2,4\},&\{5,6\},&\{1,3\}\\
\{0,1,4\},&\{3,5\},&\{2,6\}
  \end{array}
\]
In fact, this resolvable set system can be found in \cite[Example
2.3]{DanzigerStevens:2001} as a {\em class-uniformly resolvable design}
with partition $2^23^1$.

\end{proof}

\begin{theorem}\label{q3a7}$n_{3}(7,T)=\left\lceil\frac{21}{8}T\right\rceil$ for all $T\geq 1$.
\end{theorem}
\begin{proof}For $a=7$, we have $GBT_3(7,T)=\left\lceil\frac{21}{8}T\right\rceil$. When $1\leq T\leq 7$,
$n_{3}(7,T)=GBT_3(7,T)$ is known by \cite{NaydeTor2009}. By
Lemma~\ref{q3t8a7}, $n_{3}(7,8)=GBT_3(7,8)$. Hence the construction
$C_T=C_8|C_{T-8}$ gives the optimal code of length
$\left\lceil\frac{21}{8}T\right\rceil$ for all $T$  by Lemma~\ref{recur}.
\end{proof}

In \cite{NaydeTor2009}, the authors stated that
$n_3(9,1)=4=GBT_3(9,1)+1$ and $n_3(9,4)=12=GBT_3(9,4)+1$ by computer
search. However, the latter is not true. In fact, we find a
 $3$-EC-AUED code of length $11$ and with bigger size $12$, which meets the
bound in Theorem~\ref{lbound}.

\begin{lemma}
\label{q3t4a12}$n_{3}(a,4)=11$ for $8\leq a\leq 12$.
\end{lemma}
\begin{proof} For $8\leq a\leq 12$, we have $GBT_{3}(a,4)=11$.  A
$3$-EC-AUED code of size $12$ and length $11$ is constructed below.
For $8\leq a\leq 11$, the optimal codes are obtained by collecting
any set of $a$ codewords from $C_4$.
\begin{align*}
C_4=\begin{bmatrix}10022021012\\
00202112021\\
20100211202\\
01020202211\\
12001122200\\
11111111111\\
21221100002\\
02112002102\\
12210201020\\
21012010220\\
22200020111\\
00121220120
\end{bmatrix}
\end{align*}\end{proof}

Hence by Lemmas~\ref{recur} and \ref{known}, we determine all values
of $n_{3}(8,T)$ and $n_{3}(9,T)$.
\begin{theorem}\label{q3a9}$n_{3}(8,T)=n_{3}(9,T)=\left\lceil\frac{8}{3}T\right\rceil$ for all $T>1$.
\end{theorem}

Next, we  study values of $n_3(a,T)$ for $a\in \{10,11,12\}$.
By Lemma~\ref{deleter}, it is enough to consider the cases
$a=10,12$, for which we have  $GBT_3(10,T)= \left\lceil
\frac{30T}{11}\right\rceil$ and $GBT_3(12,T)= \left\lceil \frac{11T}{4}\right\rceil$.

\begin{lemma}
\label{q3t2a16}$n_{3}(a,2)=6$ for $7\leq a\leq 16$ and
$n_{3}(a,3)=9$ for $10\leq a\leq 25$.
\end{lemma}
\begin{proof} For $7\leq a\leq 16$, we have $GBT_{3}(a,2)=6$, while for $10\leq a\leq 25$, we have $GBT_{3}(a,3)=9$.
An optimal $1$-EC-AUED code of size $16$ and a $2$-EC-AUED code of
size $25$ are listed below. Optimal codes with smaller sizes can be
obtained by deleting some codewords from $C_2$ and $C_3$
respectively.
\begin{align*}
C_2=\begin{bmatrix}
211002\\
202011\\
201120\\
100221\\
112020\\
120012\\
010122\\
021021\\
022110\\
012201\\
102102\\
111111\\
121200\\
210210\\
220101\\
001212
\end{bmatrix},\
&C_3=\begin{bmatrix}
010220211\\
211200102\\
220021101\\
102020112\\
122100201\\
111111111\\
110012202\\
001102212\\
022011210\\
012121002\\
200211012\\
221120010\\
202112100\\
021212001\\
120202110\\
100122021\\
002201121\\
011022120\\
020110122\\
112210020\\
121001022\\
201010221\\
210101220\\
212002011\\
101221200
\end{bmatrix}.
\end{align*}
\end{proof}

By now, we have determined $n_{3}(12,T)$ for $2\leq T\leq 4$. Since $n_3(9,1)=4$, we have $n_{3}(12,1)\geq 4$ which is bigger than
$GBT_{3}(12,1)=3$. By
Lemma~\ref{recur}, we still need to determine $n_{3}(12,5)$ to
construct all optimal codes of size $12$.

\begin{lemma}
\label{q3t5a12}$n_{3}(a,5)=14$ for $10\leq a\leq 12$.
\end{lemma}
\begin{proof}For $10\leq a\leq 12$, we have $GBT_{3}(a,5)=14$. An
optimal code of size $12$ is given below.
\begin{align*}
C_5=\begin{bmatrix}01212201221000\\
22002120022001\\
01000222102220\\
10201101112211\\
11021100201122\\
21112010101211\\
20022021210110\\
22220002000202\\
12120111120020\\
10110222011012\\
02211210010121\\
00101012222102
\end{bmatrix}
\end{align*}
\end{proof}

Hence by Lemmas~\ref{recur} and \ref{known}, we determine all values
of $n_{3}(12,T)$.
\begin{theorem}
\label{q3a12}$n_{3}(11,T)=n_{3}(12,T)=\left\lceil \frac{11T}{4}\right\rceil$
for all $T>1$.
\end{theorem}

By simple computation, we know that $GBT_3(10,T)=GBT_3(12,T)$ for
most integers $T$. The smallest integer $T$ with
$GBT_3(10,T)<GBT_3(12,T)$ is $11$.

\begin{lemma}
\label{q3t11a10}$n_{3}(10,11)=30$.
\end{lemma}
\begin{proof}Let $X=\bbZ_{10}$. The following nine blocks form three
parallel classes in each row. Develop them to $30$  parallel classes
by $\bbZ_{10}$ and keep the order of blocks in each parallel class.
Then one can check this resolvable set system gives a ternary
$10$-EC-AUED code of size $10$ and length $30$, which is optimal by
Theorem~\ref{lbound}. In fact, the codewords are the rows of the $10 \times 30$ matrix $M = (A|B|C)$ where
$A$, $B$ and $C$ are circulant matrices with rows indexed by $\bbZ_{10}$. The leftmost column of $A$ has zeroes in rows $0,1,2,3$, ones in rows $4,6,8$, and twos in rows $5,7,9$; the leftmost column of $B$ has zeroes in rows $0,1,4,5$, ones in rows $2,6,9$, and twos in rows $3,7,8$; and the leftmost column of $C$ has zeroes in rows $0,2,3,7$, ones in rows $1,6,9$, and twos in rows $4,5,8$.
\[\begin{array}{lll}
    \{0,1,2,3\},&\{4,6,8\},&\{5,7,9\}\\
\{0,1,4,5\},&\{2,6,9\},&\{3,7,8\}\\
\{0,2,3,7\},&\{1,6,9\},&\{4,5,8\}
  \end{array}
\]
\end{proof}

Now we are in a position to determine $n_{3}(10,T)$ for all $T>1$.

\begin{theorem}
\label{q3a10}$n_{3}(10,T)=\left\lceil \frac{30T}{11}\right\rceil$ for all
$T>1$.
\end{theorem}
\begin{proof}For $a=10$, we have $GBT_3(10,T)=\left\lceil
\frac{30T}{11}\right\rceil$. For $T=2,\ldots,10,12$, the bound is met
since $GBT_3(10,T)=GBT_3(12,T)$ for these cases. For $T=11$, the
bound is achieved by Lemma~\ref{q3t11a10}. The optimal code $C_T$
for all $T\geq 13$ is given by recursion $C_T=C_{11}|C_{T-11}$.
\end{proof}

Finally, for completeness, when $T=1$, we note that
$n_{3}(a,1)=4=GBT_{3}(a,1)+1$ for $8\leq a \leq 19$ since by de
Bruijn et al. \cite{Bruijn:1951},
\[B(n,q)=\left\{(x_1,x_2,\ldots,x_n)\in \bbF_q^n|\sum_{i=1}^{n}x_i=\left\lceil\frac{n(q-1)}{2}\right\rceil\right\}\]
is a $0$-EC-AUED code with maximal size for given length $n$.

\section{Conclusion}\label{con}

We investigated the length of the shortest $q$-ary $t$-EC-AUED codes
of size $a$. A direct construction of optimal codes was
given via one-factorizations of complete graphs. We further provided a general construction of a $(d-1)$-EC-AUED code of length $2n$ from a code of length $n$ and minimum Hamming distance $d$. Finally, we would like to suggest the study of codes for which the words are the rows of a concatenation
of circulant matrices, similar to those constructed from resolvable packings in Section V.

\section*{Acknowledgments}

The authors thank the anonymous reviewers  and the associate editor for their constructive
comments and suggestions that greatly improved the quality of
this paper.

\vskip 10pt

\begin{IEEEbiographynophoto}{Yeow Meng Chee}
(SM'08) received the B.Math. degree in computer science
and combinatorics and optimization and the M.Math. and Ph.D. degrees in
computer science from the University of Waterloo, Waterloo, ON, Canada, in
1988, 1989, and 1996, respectively.

Currently, he is a Professor at the Division of Mathematical Sciences,
School of Physical and Mathematical Sciences, Nanyang Technological University,
Singapore. Prior to this, he was Program Director of Interactive Digital
Media R\&D in the Media Development Authority of Singapore, Postdoctoral
Fellow at the University of Waterloo and IBM¡¯s Z\"urich Research Laboratory,
General Manager of the Singapore Computer Emergency Response Team,
and Deputy Director of Strategic Programs at the Infocomm Development
Authority, Singapore.

His research interest lies in the interplay between combinatorics and computer
science/engineering, particularly combinatorial design theory, coding
theory, extremal set systems, and electronic design automation.
\end{IEEEbiographynophoto}
\begin{IEEEbiographynophoto}{Xiande Zhang}
received the Ph.D. degree in mathematics from Zhejiang University, Hangzhou, Zhejiang, P. R. China in 2009. From 2009 to 2015, she held postdoctoral positions in Nanyang Technological University and Monash University. Currently, she is a Research Professor at  school of Mathematical Sciences, University of Science and Technology of China.  Her research interests include combinatorial design theory, coding theory, cryptography, and their interactions.  She received the 2012 Kirkman Medal from the Institute of Combinatorics and its Applications.
\end{IEEEbiographynophoto}

\end{document}